\documentclass[journal]{IEEEtran}
\usepackage{cite}
\usepackage{amsmath,amssymb,amsfonts}
\usepackage{algorithmic}
\usepackage{graphicx}
\usepackage{epstopdf}
\usepackage{algorithm}
\usepackage{textcomp}
\usepackage{xcolor}
\usepackage{amsthm}
\usepackage{float} 
\usepackage{subcaption}
\newtheorem{Theorem}{Theorem}
\IEEEpubid{\begin{minipage}{\textwidth}\ \centering
        Copyright (c) 2025 IEEE. Personal use of this material is permitted. \\
        However, permission to use this material for any other purposes must be obtained from the IEEE by sending a request to pubs-permissions@ieee.org.
\end{minipage}}

\begin{document}
\title{Low Complexity Weight Flexible Decoding Schemes of Linear Block Code for 6G xURLLC}

\author{Di Zhang,~\IEEEmembership{Senior Member,~IEEE,}
        Yinglei Yang,
        Zilong Liu,~\IEEEmembership{Senior Member,~IEEE,}
        Shaobo Jia,
        Kyungchun Lee,~\IEEEmembership{Senior Member,~IEEE,}
        and~Zhirong Zhang

\thanks{Di Zhang is with the School of Intelligent Systems Engineering, Sun Yat-sen University, Shenzhen 518107, China (E-mail: zhangd263@mail.sysu.edu.cn).}

\thanks{Yinglei Yang and Shaobo Jia are with the School of Electrical and Information Engineering, Zhengzhou University, Zhengzhou 450001, China (E-mail: yingleiyang@gs.zzu.edu.cn, ieshaobojia@zzu.edu.cn).}

\thanks{Zilong Liu is with the School of Computer Science and Electronics Engineering, University of Essex, Colchester CO4 3SQ, U.K. (E-mail: zilong.liu@essex.ac.uk).}
\thanks{Kyungchun Lee is with the Department of Electrical and Information Engineering and the Research Center for Electrical and Information Technology, Seoul National University of Science and Technology, Seoul  01811, Republic of Korea (E-mail: kclee@seoultech.ac.kr).}
\thanks{Zhirong Zhang is with the China Telecom Research Institute of Mobile and Terminal Technology, Beijing 102209, China (E-mail: zhangzhr@chinatelecom.cn).}
}

\markboth{
}
{Shell \MakeLowercase{\textit{et al.}}: Bare Demo of IEEEtran.cls for IEEE Journals}
\maketitle


\begin{abstract}
Low complexity error correction code is a key enabler for next generation ultra-reliable low-latency communications (xURLLC) in six generation (6G). Against this background, this  paper proposes a  decoding scheme for linear block code by leveraging certain interesting properties of  dual codewords. It is found  that dual codewords with flexible weights can provide useful decoding information for the locations and magnitudes of error bits, which yielding higher reliability performance. In addition, two decoding schemes are proposed, in which one directly utilizes intrinsic information for iterative decoding, and the other combines prior channel information with intrinsic information for decoding. Both schemes are implemented using vector multiplication and real-number comparisons, making them easy to implement in hardware. Simulation results demonstrate the validness of our study. 
\end{abstract}

\IEEEpeerreviewmaketitle

\section{Introduction}

A remarkable transition from fifth generation (5G) to six generation (6G) is the expansion of ultra-reliablie and low latency communications (URLLC) into next generation URLLC (xURLLC)\cite{Lulu}. Use cases motivating this shift include a broad range of mission-critical applications that were beyond 5G’s reach. For example, in industrial automation and intelligent transportation applications, even a single packet loss can have serious consequences, which cannot be accomplished by 5G's URLLC technologies. Closing the gap (e.g., system sum-rate and quality-of-service\cite{AsifRIS,ZainAliNOMA}) between 5G URLLC and 6G’s diverse xURLLC targets thus calls for novel physical-layer solutions. Short codes will be important as machine-type communication (MTC) will become a major driver of 6G traffic \cite{6Gmiao}.  To meet the reliability error rate targets, powerful error-correcting codes that can deliver near-zero error rates even at short block lengths is an ideal choice\cite{Jingjing}. This necessity naturally points to linear block code, which combine strong error correction capability with efficient implementation. While linear coding boosts reliability at the transmitter side, it is only half of the equation for next generation reliable communication. The other half lies in decoding, the receiver’s ability to correct errors within extremely tight time and limited computational ability\cite{Maxiao,YangPLDPC}. 

In literature, decoding algorithms for linear block code, such as maximum likelihood decoding (MLD) and minimum distance decoding (MDD), exhibit distinct advantages and limitations under various application scenarios\cite{richardson2008modern}. With the rediscovery of low-density parity-check (LDPC) codes\cite{gallager1962low}, iterative decoding algorithms with the aid of belief propagation (BP)\cite{BP_VTC}, have attracted many research attentions. Although BP leads to reduced multiplication operations\cite{Zhu_TWC}, the introduction of $\tanh$ and $\tanh^{-1}$ functions increases computational complexity. To circumvent  this, minimum sum decoding (MSD) algorithm was proposed in \cite{fossorier1999reduced}  by replacing  the $\tanh$ and $\tanh^{-1}$ functions with straightforward sign evaluations and numerical comparisons. However, the complexity are under practical constraints, such as finite processing delays and high spectral efficiency, continuing to pose a significant challenge, especially in xURLLC  scenarios \cite{9306872, Jia_TWC, Mingxiao}. To reduce the decoding complexity, Bossert introduced a method in \cite{bossert1986hard} that uses the minimum weight codewords of dual codes. Such a decoding scheme was further extended for Bose–Chaudhuri–Hocquenghem (BCH) codes in \cite{bossert2022hard}.  New shift-sum decoding methods for non-binary cyclic codes were proposed in \cite{yuan2021plausibility,xing2023shift} by exploiting the statistical distribution of frequency matrix\footnote{The frequency matrix counts the occurrence frequency of the coefficients of the syndrome polynomials at different positions in the vector.}. In addition to that, artificial intelligence (AI) was recently introduced to linear block code design for high performance 6G with affordable complexity\cite{ai_tvt}. 
\IEEEpubidadjcol

 However, the aforementioned research results are mostly based on the polynomial operations of cyclic codes. Besides, with the minimum weight codewords of dual codes, one may not be able to extract more reliability information due to the limited number of minimum weight dual codewords. To solve this issue, we propose a decoding method using weight-unconstrained dual codewords in this article. The main contributions are summarized as follows:
\begin{itemize}
\item We show that most dual codewords with diverse weights offer effective reliability information for evaluating the weights of error vectors during decoding. Specifically, the larger the weight of the error vector, the higher the corresponding reliability information value. The reliability requirement of 6G can thus be achieved via this method. 

\item We propose two low complexity decoding algorithms named iterative error reduction decoding (IERD) and the pre-knowledge assisted decoding (PAD), which are highly efficient for hardware implementation. The IERD can precisely identify the least reliable bit during each iteration, and PAD further optimizes the decoding process by incorporating channel information, thus enabling a multi-bit parallel flipping mechanism. 

\end{itemize}

\section{Prerequisites}
In the $n$-dimensional vector space $\mathbb{F}_q^n$, if a non-empty subset $\boldsymbol{C} \subseteq \mathbb{F}_q^n$ forms an $\mathbb{F}_q$-linear subspace of $\mathbb{F}_q^n$, then $\boldsymbol{C}$ is referred to as a $q$-ary linear block code. Any vector $\mathbf{c}=\left(c_1, c_2, \cdots, c_n\right)$ within $\boldsymbol{C}$ is termed a codeword. Here, $n$ is known as the length of $\boldsymbol{C}$, $K=|\boldsymbol{C}|$ is the number of codewords, $k=\log_q K$ is the information content of $\boldsymbol{C}$, and $\frac{k}{n}$ is called the code rate of $\boldsymbol{C}$. For a codeword $\mathbf{c}$, its Hamming weight is defined as $\text{wt}(\mathbf{c}) = \sum_{i=1}^{n} c_i$.
For two vectors $\mathbf{c}^1, \mathbf{c}^2\in \mathbb{F}_q^n$, the Hamming distance between the codewords $\mathbf{c}^1$ and $\mathbf{c}^2$ is denoted by $d_H(\mathbf{c}^1, \mathbf{c}^2) = \sum_{i=1}^{n} |c^{1}_{i} - c^{2}_{i}|$. 

For a $q$-ary linear block code $\boldsymbol{C}$, its minimum Hamming distance is defined as
\begin{align}
d=d(\boldsymbol{C})=\min \left\{d(\mathbf{c}^1, \mathbf{c}^2): \mathbf{c}^1, \mathbf{c}^2 \in \boldsymbol{C}, \mathbf{c}^1 \neq \mathbf{c}^2\right\}.
\end{align}
 Let $n \in \mathbb{Z}^{+}$, the inner product of two vectors $\boldsymbol{x}=\left(x_1, \ldots, x_n\right)$ and $\boldsymbol{y}=\left(y_1, \ldots, y_n\right)$ in $\mathbb{F}_q^n$ is defined as $\langle\boldsymbol{x}, \boldsymbol{y}\rangle=\sum_{i=1}^n x_i y_i$.
Then, for a $q$-ary linear block code $\boldsymbol{C}$ of length $n$ over $\mathbb{F}_q$, its dual code is $\boldsymbol{C}^{\perp}=\left\{\boldsymbol{x} \in \mathbb{F}_q^n \mid \langle\boldsymbol{x},\boldsymbol{c}\rangle\equiv 0 \pmod{q}, \forall \boldsymbol{c} \in \boldsymbol{C}\right\}$. In this paper, we focus on linear block code in the binary case.

\section{Decoding Based On The Dual Codewords}
 
 The significance of an error vector \(\boldsymbol{f}\) is typically assessed by its weight, defined by the number of non-zero elements, \(wt(\boldsymbol{f})\). For a codeword \(\boldsymbol{c} \in \boldsymbol{C}\) and its dual codeword \(\boldsymbol{v} \in \boldsymbol{C}^{\perp}\), it is known that their inner product satisfies \(\langle\boldsymbol{v}, \boldsymbol{c}\rangle \equiv 0 \pmod{2}\). For any error vector $\mathbf{f}$, in most cases, the inner product between the receive vector \(\boldsymbol{r} = \boldsymbol{c} + \boldsymbol{f}\) and the dual codeword \(\boldsymbol{v}\) results in \(\langle \boldsymbol{v}, \boldsymbol{r} \rangle = \langle \boldsymbol{v}, \boldsymbol{c} + \boldsymbol{f} \rangle = \langle \boldsymbol{v}, \boldsymbol{f} \rangle \equiv 1 \pmod{2}\). This relationship forms the basis for verifying whether the receive vector \(\boldsymbol{r}\) is from the transmit codeword. However, the reliability of this verification, as well as how to quantitatively measure this reliability, remains a critical question. In the following, we will derive and analyze this issue in detail.
 
 
 Define the weight of the error vector \(\boldsymbol{f}\) as \(\tau\), and consider the case where the odd number of erroneous positions in \(\boldsymbol{f}\) overlaps with the \(\delta\) non-zero positions of the dual codeword \(\boldsymbol{v}\). In this situation, \(\langle\boldsymbol{v}, \boldsymbol{f}\rangle \equiv 1 \pmod{2}\). 
 We first consider the probability that \(k\) non-zero positions of error vector \(\boldsymbol{f}\) partially overlap with \(\delta\) non-zero positions of \(\boldsymbol{v}\), which can be expressed as the product of the probability that \(\delta\) non-zero positions overlap with \(k\) error positions and the probability that \(N-\delta\) zero positions overlap with \(\tau-k\) non-zero positions. Specifically, it becomes  
\begin{align}
Pr(k,\tau,\delta) &= Pr_{k}(k,\tau,\delta) \cdot Pr_{\tau-k}(k,\tau,\delta) \notag\\
&=  \frac{\binom{\tau}{k}\cdot\frac{\delta!}{(\delta-k)!} \cdot \frac{(N-\delta)!}{(N-\delta-\tau+k)!}}{\frac{N!}{(N-\tau)!}}.
\end{align}

Then, for all \(\boldsymbol{v} \in C^{\perp}\) with \(wt(\boldsymbol{v}) = \delta\) and \(wt(\boldsymbol{f}) = \tau\), the expected probability that \(\langle\boldsymbol{v}, \boldsymbol{f}\rangle \equiv 1 \pmod{2}\) is

\begin{align}
\label{eq:W}
W(\delta,\tau) &= \sum_{\substack{k=1 \\ k \text{ odd}}}^{\tau} \frac{\binom{\tau}{k}\cdot\frac{ \delta!}{(\delta-k)!} \cdot \frac{(N-\delta)!}{(N-\delta-\tau+k)!}}{\frac{N!}{(N-\tau)!}} 
= \sum_{\substack{k=1 \\ k \text{ odd}}}^{\tau} \frac{\binom{\tau}{k} \binom{N-\tau}{\delta-k}}{\binom{N}{\delta}}.
\end{align}

In the following theorem, we identify the specific condition that \( W(\delta, \tau) \) increases monotonically with respect to \( \tau \).

\begin{Theorem}
\label{thm1}
    For code length $ N $, the expected probability $ W(\delta, \tau) $ is a function defined on the parameters $ \delta $ and $ \tau $. If $2\tau +2 +(\sqrt{\tau}+1)(\delta-3) \leq N$, then $ W(\delta, \tau) $ is monotonically increasing with respect to $ \tau $.
\end{Theorem} 
\begin{proof}
    See Appendix A.
\end{proof}

Moreover, We can derive the following equation from (\ref{eq:W}).  
\begin{align}
\label{eq:WN}
W(N-\delta,\tau)
&= \sum_{\substack{k=1 \\ k \text{ odd}}}^{\tau} \frac{\binom{\tau}{k} \binom{N-\tau}{N-\delta-k}}{\binom{N}{N-\delta}}= \sum_{\substack{k=1 \\ k \text{ odd}}}^{\tau} \frac{\binom{\tau}{k} \binom{N-\tau}{\delta-(\tau-k)}}{\binom{N}{\delta}}\notag\\
&=\begin{cases}
 1 - W(\delta,\tau) & \text{if } \tau \equiv 1 \pmod{2}, \\
 W(\delta,\tau) & \text{if } \tau \equiv 0 \pmod{2}.
\end{cases}
\end{align}

\begin{figure*}[htbp]
    \centering
    \begin{subfigure}{0.32\textwidth}
        \includegraphics[width=\linewidth,height=0.17\textheight]{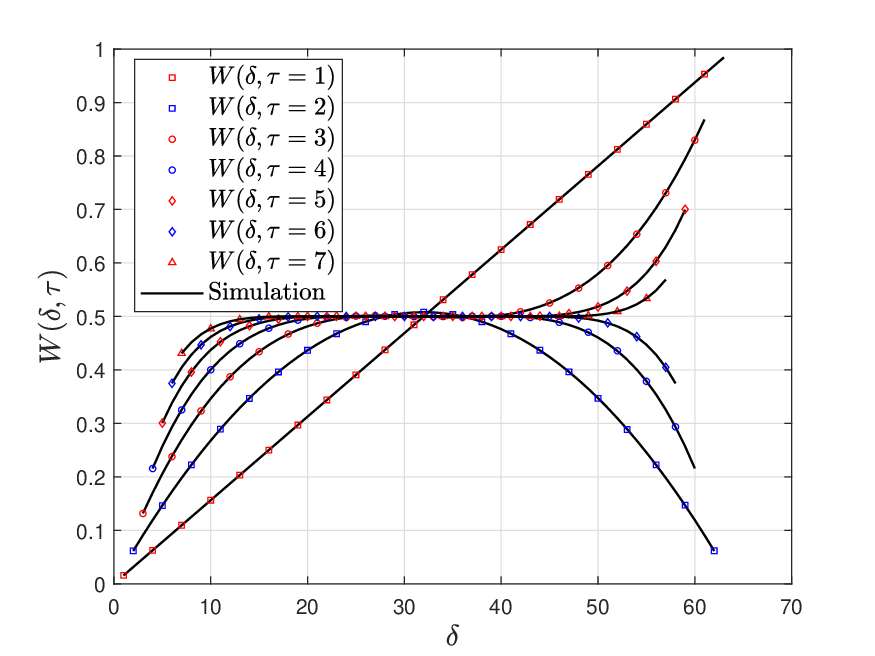}
        \caption{ Expected probability in linear block code $\mathbf{C}(64, 16)$.}
        \label{fig}
    \end{subfigure}
     \hfill
    \begin{subfigure}{0.32\textwidth}
        \includegraphics[width=\linewidth,height=0.17\textheight]{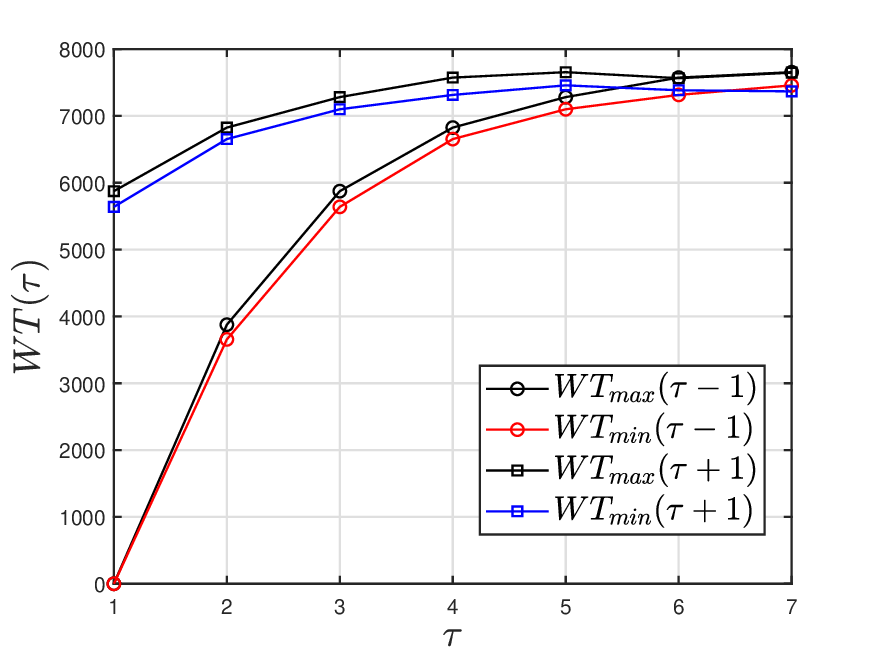}
        \caption{Intrinsic information in linear block code $\mathbf{C}(32, 16)$.}
        \label{fig2}
    \end{subfigure}
    \hfill
    \begin{subfigure}{0.32\textwidth}
        \includegraphics[width=\linewidth,height=0.17\textheight]{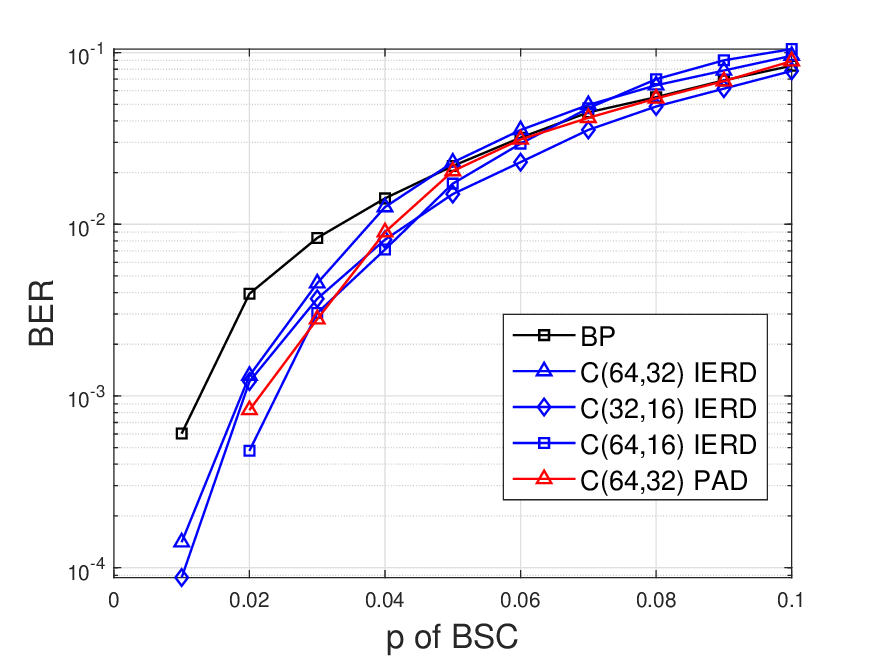}
        \caption{ BER performance of $\mathbf{C}(n,k)$-linear block code under BSC.}
        \label{fig3}
    \end{subfigure}
    \caption{Decoding performances based on dual codewords.}
    \label{fig:main}
\end{figure*}

Error vectors with different weights yield different expected probabilities $W(\delta, \tau)$, as shown in Fig. \ref{fig:main}(\subref{fig}).  When $ wt(\boldsymbol{v}) < d_{A} \leq   \frac{N-2\tau+3\sqrt{\tau}+1}{\sqrt{\tau}+1}$, the higher the weight of $wt(\mathbf{f})$, the higher the expected probability $W(\delta, \tau)$. Moreover, $W(\delta, \tau)$ can  effectively distinguish the weight values of the error vectors, which can serve as a reliability metric for determining the magnitude of the error vector. Based on (\ref{eq:WN}), while calculating $W(\delta, \tau)$ for high-weight dual codewords using 
\begin{align}
W(\delta,\tau)=\begin{cases}
 1 - W(\delta,\tau) & \text{if } \tau \equiv 1 \pmod{2}, \\
 W(\delta,\tau) & \text{if } \tau \equiv 0 \pmod{2},
\end{cases}
\end{align}
we find it increasing with the increase weight $\tau$ of random errors. This result implies that most dual codewords (when $wt(\boldsymbol{v})<d_{A} $ or $ wt(\boldsymbol{v})>d_{B} \geq N-\frac{N-2\tau+3\sqrt{\tau}+1}{\sqrt{\tau}+1}$) can provide useful information for decoding. When $d_{A} < wt(\boldsymbol{v}) < d_{B}$, the expected probability $|W(\delta, \tau)-0.5| \leq 10^{-3}$ for all $ \tau \in \mathbb{Z}_{\geq 0}$. Therefore, the dual codewords in this range neither provide any useful information nor generate any detrimental information in the decoding process. With the analysis above, one can present the specific theory of the dual codewords decoding as follows.

Let \( \mathbf{C} \subseteq \mathbb{F}_2^n \) be a linear block code, and its dual code is defined as \( \mathbf{C}^{\perp} := \left\{\boldsymbol{v} \in \mathbb{F}_2^n \mid \forall_{\boldsymbol{c} \in \mathbf{C}}\langle\boldsymbol{v}, \boldsymbol{c}\rangle\equiv 0 \pmod{2}\right\} \). Within the dual code \( \mathbf{C}^{\perp} \), we define  set \( \mathbf{B} \) as \(  \left\{\boldsymbol{b} \in \mathbf{C}^{\perp} \mid w t(\boldsymbol{b})>d_{B}\geq  N-\frac{N-2\tau+3\sqrt{\tau}+1}{\sqrt{\tau}+1} \right\}=\left\{\boldsymbol{b}_1, \cdots, \boldsymbol{b}_s\right\} \subseteq \mathbf{C}^{\perp} \), where \( |\mathbf{B}| \) denotes the cardinality of  set \( \mathbf{B} \), and the elements of set \( \mathbf{B} \) are referred to as decoding vectors. Similarly, we define set \( \mathbf{A} \) as \( \left\{\boldsymbol{a} \in \mathbf{C}^{\perp} \mid w t(\boldsymbol{a})<d_{A}\leq \frac{N-2\tau+3\sqrt{\tau}+1}{\sqrt{\tau}+1}\right\}=\left\{\boldsymbol{a}_1, \cdots, \boldsymbol{a}_s\right\} \subseteq \mathbf{C}^{\perp} \), where \( |\mathbf{A}| \) denotes the number of elements in \( \mathbf{A} \), and the elements of set \( \mathbf{A} \) are also referred to as decoding vectors. 

Given the received codeword \(\boldsymbol{w} = \boldsymbol{c} + \boldsymbol{f}\), where \(\boldsymbol{c} \in \mathbf{C}\) and \(\boldsymbol{f}\) is the error vector, the intrinsic information \cite{bossert2022hard} \(WT\) is defined as \(WT = W T_A(\boldsymbol{w}) + W T_B(\boldsymbol{w})\). Here \(W T_A(\boldsymbol{w})\) represents the intrinsic information obtained from the elements of \(\mathbf{A}\). The calculation of \(W T_A(\boldsymbol{w})\) can be expressed as follows
\begin{align}
W T_A(\boldsymbol{w})
=\sum_{\boldsymbol{a} \in \mathbf{A}} W(wt(\boldsymbol{a}),wt(\boldsymbol{w}))
=\sum_{\boldsymbol{a} \in \mathbf{A}} \langle\boldsymbol{a}, \boldsymbol{w}\rangle \mod(2).
\end{align}
Here \(W T_B(\boldsymbol{w})\) represents the intrinsic information that can be obtained from the elements in \(\mathbf{B}\). Therefore, 
\begin{align}
     W T_B(\boldsymbol{w}) 
    =\begin{cases}
        |\mathbf{B}|-\sum_{\boldsymbol{b} \in \mathbf{B}} (\langle\boldsymbol{b}, \boldsymbol{w}\rangle) &  w t(\boldsymbol{f})\equiv 1 \pmod{2},\\
        \sum_{\boldsymbol{b} \in \mathbf{B}} (\langle\boldsymbol{b}, \boldsymbol{w}\rangle ) &  w t(\boldsymbol{f})\equiv 0 \pmod{2}.
    \end{cases}
\end{align}

Based on the definition of the dual code, we can deduce that \( W T_B(\boldsymbol{w})=W T_B(\boldsymbol{c}+\boldsymbol{f}) =W T_B(\boldsymbol{f}) \), and \( W T_A(\boldsymbol{w})=W T_A(\boldsymbol{c}+\boldsymbol{f}) =W T_A(\boldsymbol{f}) \). Moreover, both $W T_A(\boldsymbol{f})$ and $W T_B(\boldsymbol{f})$ are monotonically increasing with respect to $wt(\boldsymbol{f})$. Thus, for error vectors \(\boldsymbol{f}\) and \(\boldsymbol{h}\), if \(wt(\boldsymbol{f}) < wt(\boldsymbol{h})\), it can be concluded that \(W T(\boldsymbol{f}) < W T(\boldsymbol{h})\).  Next, we will use one example to further illustrate how intrinsic information is extracted and how it aids in decoding.

 Example 2 (Intrinsic information extraction and its role in decoding): Taking the linear block code $\mathbf{C}(32,16)$  as an example, we conducted simulation experiments using 5000 dual codes from sets A and B, respectively. In the experiments, random errors $\boldsymbol{f}$ with weight $\tau$ were simulated. For a given received codeword \(\boldsymbol{w} = \boldsymbol{c} + \boldsymbol{f}\), the intrinsic information is computed as  $WT_i = W T_B(\boldsymbol{w} + \mathbf{e}_i) + W T_A(\boldsymbol{w} + \mathbf{e}_i)$,   where \(\mathbf{e}_i\) denotes the vector whose \(i\)-th component is 1 and all other components are 0. As shown in  Fig. \ref{fig:main}(\subref{fig2}), the minimum of the intrinsic information $WT(wt(\boldsymbol{f})=\tau-1)$ obtained by flipping error positions is less than the minimum obtained by flipping correct positions. In this situation, error positions and correct positions can be distinctly identified (for example, $WT_{\min} = \min(WT)$ corresponds to the flipped error positions), thus implying a decoding philosophy based on unreliable positions.





\section{Low Complexity Decoding Schemes with Hard Decision}
In this section, we introduce two low complexity decoding Schemes and analyze their performances. 
\subsection{The Proposed Decoding Schemes}
For all decoding schemes, the first step is to randomly generate the sets \( \mathbf{A} \) and \( \mathbf{B}\) of the relevant dual codewords.
At the receiver, we calculate the \( WT \) and sort the obtained \( WT \) values based on their reliabilities.
Since one may not know in advance the weight of the error vector \(\boldsymbol{f}\), we make the following adjustment to \(WT_B\):
\begin{align}
     W T_B(\boldsymbol{w})
    =\begin{cases}
        |B|-\sum_{\boldsymbol{b} \in \mathbf{B}} (\langle\boldsymbol{b}, \boldsymbol{w}\rangle) &  \sum_{\boldsymbol{b} \in \mathbf{B}} (\langle\boldsymbol{b}, \boldsymbol{w}\rangle)>\frac{|\mathbf{B}|}{2},\\
        \sum_{\boldsymbol{b} \in \mathbf{B}} (\langle\boldsymbol{b}, \boldsymbol{w}\rangle ) &  \sum_{\boldsymbol{b} \in \mathbf{B}} (\langle\boldsymbol{b}, \boldsymbol{w}\rangle)<\frac{|\mathbf{B}|}{2}.
    \end{cases}
\end{align}

            
1) IERD Scheme: Iterative error reduction decoding is a hard-decision decoding method, where in each iteration, the unreliable positions based on \( WT \) are flipped to compute a new received vector, as detailed in \textbf{Algorithm \ref{alg:AOA}}. In addition, $WT(\mathbf{f}+\mathbf{e}_i)=\min(WT)$ corresponds to the possibly largest $\mathbf{e}_i$ such that $wt(\mathbf{f}+\mathbf{e}_i)=\tau-1$. Based on this likelihood, we can calculate the most probable vector by reversing these positions to iteratively reduce the number of errors. When \( \min(WT) > 0 \), we take \( \mathbf{f} = \mathbf{f}+\mathbf{e}_i \) as the new error vector and continue to iterate. When \( \min(WT) = 0 \), this indicates that the correct transmit codeword has been found. IERD can identify the location of an error in each iteration, but it requires calculating \( WT \) every time.
\begin{algorithm}
    \caption{The IERD Algorithm}
    \label{alg:AOA}
    \renewcommand{\algorithmicrequire}{\textbf{Input:}}
    \renewcommand{\algorithmicensure}{\textbf{Output:}}
    \begin{algorithmic}[1]
        \REQUIRE Received vector $\mathbf{r}$, max iterations \( T_{max} \) 
        \ENSURE c    
        \FOR{$k=1:T_{max}$}
            \FOR{$i=1:n$}
                \STATE $WT_i=W T_B(\mathbf{r}+\mathbf{e}_i)+W T_A(\mathbf{r}+\mathbf{e}_i)$
                
            \ENDFOR
        \STATE $W T_j=\min{\boldsymbol(WT)}$
        \IF {$W T_j ==0$}
            \STATE return $\boldsymbol{r} = \mathbf{r}+\mathbf{e}_j$;  
            \STATE  break;
        \ELSE
            \STATE $\boldsymbol{r} = \mathbf{r}+\mathbf{e}_j$
        \ENDIF
        \ENDFOR
    \end{algorithmic}
\end{algorithm}

2) PAD Scheme: The pre-knowledge assisted decoding scheme builds upon \textbf{Algorithm \ref{alg:AOA}}, which uses only the intrinsic information formed by the dual codewords, by incorporating prior channel information, such as the transition probability \( p \) in binary symmetric channel (BSC). We use $LR_i = (1-p)/p$ or $LR_i = p/(1-p)$ to represent the prior information when the received symbol \( r_i = 0 \) or \( r_i = 1 \), and we express the intrinsic information as \(  \frac{WT-\min(WT)}{\max(WT)-WT} \).  In the additive white gaussian noise (AWGN) channel, the prior information can be represented by the likelihood ratios
\begin{equation}
LR_i = \frac{p(c_i = 0 \mid r_i)}{p(c_i = 1 \mid r_i) }= \frac{1 + e^{-2r_i/\sigma^2}}{1 + e^{2r_i/\sigma^2}} =e^{-2r_i/\sigma^2},
\end{equation}
where \(\sigma^2\) is the noise variance. By combining the prior information with the intrinsic information, we obtain comprehensive information
\begin{equation}
E_i=\left\{\begin{array}{l}
 \frac{LR_i(\max(WT)-WT)}{WT-\min(WT)}, r_i>0,\\
 \frac{LR_i(WT-\min(WT))}{\max(WT)-WT}, r_i<=0,
\end{array}\right.
\label{eq:Ei} 
\end{equation}
which is used to determine whether to reverse these positions so as to decrease the overall error count: if \(E_i>1\), the bit is taken as \(r_i=0\). otherwise, it is taken as \(r_i=1\). This information is further used to update the prior information $LR_i = \alpha E_i$ for the next iterations, where \(\alpha\) is a scaling factor. Compared to IERD, PAD can significantly reduce the number of calculations of intrinsic information \( WT \), thereby reducing the decoding complexity.


\begin{figure}[!t]
\centering
\includegraphics[width=3.0in]{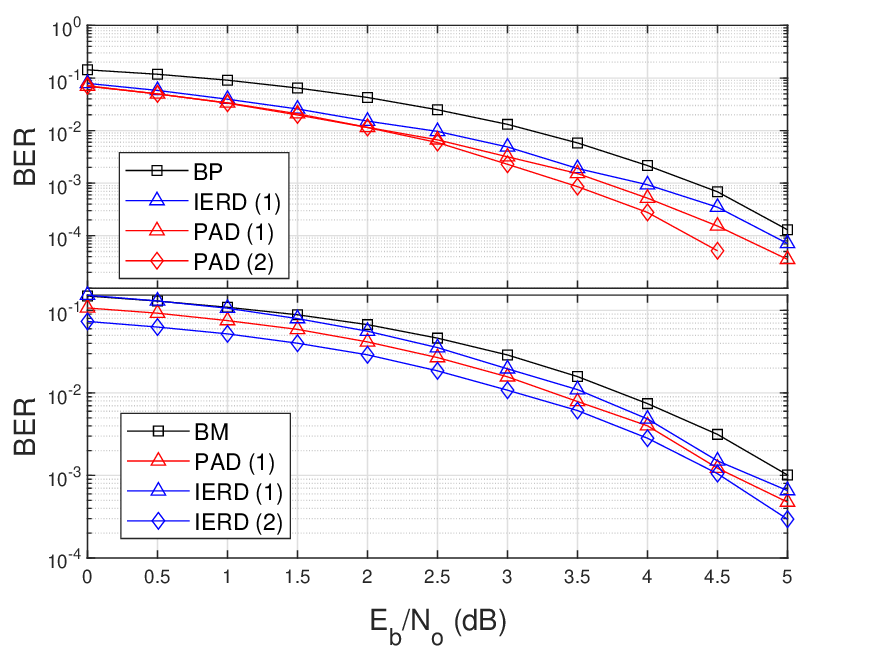} 
\caption{BER performance of $\mathbf{C}(64,22)$-linear block code and $\mathbf{C}(63,30)$-BCH code over AWGN channel.}
\label{fig4}
\end{figure}
\begin{figure}[!t]
\centering
\includegraphics[width=3.0in]{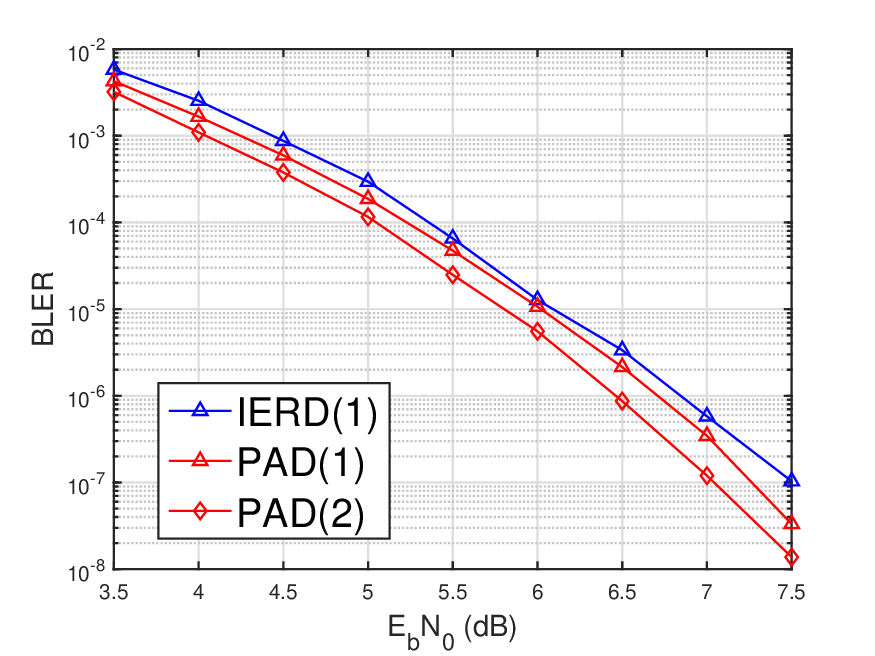} 
\caption{BLER performance of $\mathbf{C}(64,22)$-linear block code over AWGN channel.}
\label{fig5}
\end{figure}
\subsection{Error Rate Analysis}
Let the error vector be \( wt(\boldsymbol{f}) = \tau \), for all \( \boldsymbol{v} \in \mathbf{C}^{\perp} \) with \( wt(\boldsymbol{v}) = \delta \). A total of \( |\boldsymbol{d}| \) distinct dual codewords with different code weights are selected, and let \( \boldsymbol{cw} \) denote the number of dual codewords selected for each weight. Then, the probability of \( W T(\boldsymbol{\tau})=k\) can be given as
\begin{align}
& Pr(W T(\boldsymbol{\tau})=k) \notag\\
&= \sum_{ \sum_{i=1}^{|\boldsymbol{d}|} k_i = k}\prod_{i=1}^{|\boldsymbol{d}|}  \binom{\boldsymbol{\boldsymbol{cw}}(i)}{k_i} p_{\tau}(i)^{k_i} (1-p_{\tau}(i))^{\boldsymbol{\boldsymbol{cw}}(i)-k_i}.
\end{align}
where $p_{\tau}(i)=W(d(i),\tau)$. 

In IERD scheme, after flipping an error position and a non-error position once, the weights of the resulting new error vectors $\mathbf{f}_1$ and $\mathbf{f}_2$ are $\tau-1$ and $\tau+1$, respectively. The probability that the intrinsic information obtained from $\mathbf{f}_1$ is less than that obtained from $\mathbf{f}_2$ is the probability of finding and reversing the error position. Specifically, it becomes
\begin{align}
&Pr(W T({\tau+1})>W T({\tau-1})) \notag\\
&= \sum_{k=1}^{\sum(\boldsymbol{cw})}Pr(W T({\tau+1})=k)Pr(W T({\tau-1})<k).
\end{align}

\begin{figure*}
\begin{align}
\label{eq:Pr}
Pr({\tau_{sucess}})
&= Pr(wt(\boldsymbol{f})=\tau)Pr(W T({\tau+1})>W T({\tau-1}))Pr({{(\tau-1)}_{sucess}}) \notag\\
&= \binom{n}{\tau} p^\tau(1-p)^{n-\tau} \sum_{k=1}^{\sum(\boldsymbol{cw})}Pr(W T({\tau+1})=k)Pr(W T({\tau-1})<k)Pr({{(\tau-1)}_{sucess}}).
\end{align}
\end{figure*}

Under BSC, the probability that $wt(\boldsymbol{f})=\tau$ can be expressed as $Pr(wt(\boldsymbol{f})=\tau) = \binom{n}{\tau} p^\tau(1-p)^{n-\tau}$. When \( wt(\boldsymbol{f}) = \tau \), the probability of successful decoding \( Pr(\tau_{\text{success}}) \) is given by (\ref{eq:Pr}).
Then, the word error rate (WER) for a BSC with error probability $p$ can be calculated as $\operatorname{WER}(p)=1-\sum_{\tau=1}^n Pr({\tau_{sucess}})$.


\section{Numerical Results}

The numerical results are started by generating the linear block code and obtaining the corresponding sets \( \mathbf{A} \) and \( \mathbf{B} \). The iteration counts for both the IERD and PAD schemes are set to $T_{\max} = 15$. The results are then simulated over BSC and AWGN, as shown in Fig. \ref{fig:main}(\subref{fig3}) and Fig. \ref{fig4}.  We select \( |\mathbf{A}| + |\mathbf{B}| \) dual codewords to compute the intrinsic information and compare performance with BP and the Berlekamp–Massey (BM) schemes. \(|\mathbf{A}| + |\mathbf{B}| \in \{1000, 5000\}\), corresponding to IERD (1) and IERD (2), as well as PAD (1) and PAD (2) in Fig. \ref{fig4}.   Fig. \ref{fig4}, Fig. \ref{fig6} and Fig. \ref{fig:main}(\subref{fig3}) illustrates that the proposed IERD scheme, which relies solely on intrinsic information, outperforms  BP, MSD and BM schemes slightly in terms of error reduction. In Fig. \ref{fig4}, we employ a randomly generated systematic linear block code without regard to whether the cycles in the factor graph for BP decoding are even-length, nor do we enforce sparsity in the parity-check matrix. Moreover, the proposed decoding algorithm is better suited to short codes, hence the superior performance. Furthermore, the decoding performance of the PAD scheme exceeds that of IERD scheme. Besides, IERD scheme can correct one error per iteration. Therefore, once $T_{\max} > \frac{d(\mathbf{C})-1}{2}$, increasing the number of iterations $T_{\max}$ will not significantly improve the decoding performance. 

The ability of hyper reliability performance of our proposal with linear block code is afterward verified. As shown in Fig. \ref{fig5} and Fig. \ref{fig7} , both the proposed IERD and PAD can meet the hyper-reliability requirements of 6G xURLLC.   For each iteration of the PAD scheme, the required computation involves \( ((|\mathbf{A}| + |\mathbf{B}|)n + 1)n \) finite-field multiplications, resulting in an asymptotic complexity of $\mathcal{O}(((|\mathbf{A}| + |\mathbf{B}|)n + 1)n)$. With a maximum of $T_{\max}$ iterations, the worst-case complexity becomes $\mathcal{O}(T_{\max}((|\mathbf{A}| + |\mathbf{B}|)n + 1)n)$. In contrast, each iteration of IERD scheme requires \( (|\mathbf{A}| + |\mathbf{B}|)n^2 \) finite-field multiplications, resulting in an asymptotic complexity of $\mathcal{O}((|\mathbf{A}| + |\mathbf{B}|)n^2)$. Despite this, PAD scheme benefits from fewer iterations compared to IERD scheme.  As shown in Fig. \ref{fig7}, we compare the average latency of the proposed decoding scheme with other decoding algorithms when successfully transmitting 10,000 blocks of 32 bits. This also indirectly confirms that our algorithm is low-complexity and low-latency. The proposed decoding scheme involves only vector–matrix multiplications and element-wise comparisons throughout the entire decoding process, making them easy to implement in hardware. 

\begin{figure}[!t]
\centering
\includegraphics[width=3.0in]{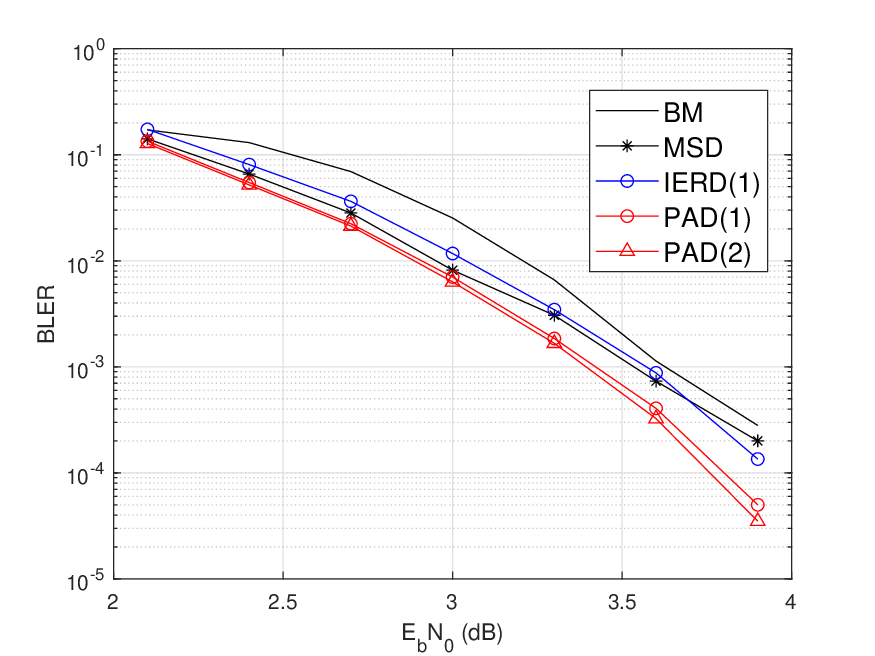} 
\caption{ BLER performance of LDPC (N = 256) and BCH (N = 255) codes over AWGN channel.}
\label{fig6}
\end{figure}
\begin{figure}[!t]
\centering
\includegraphics[width=3.0in]{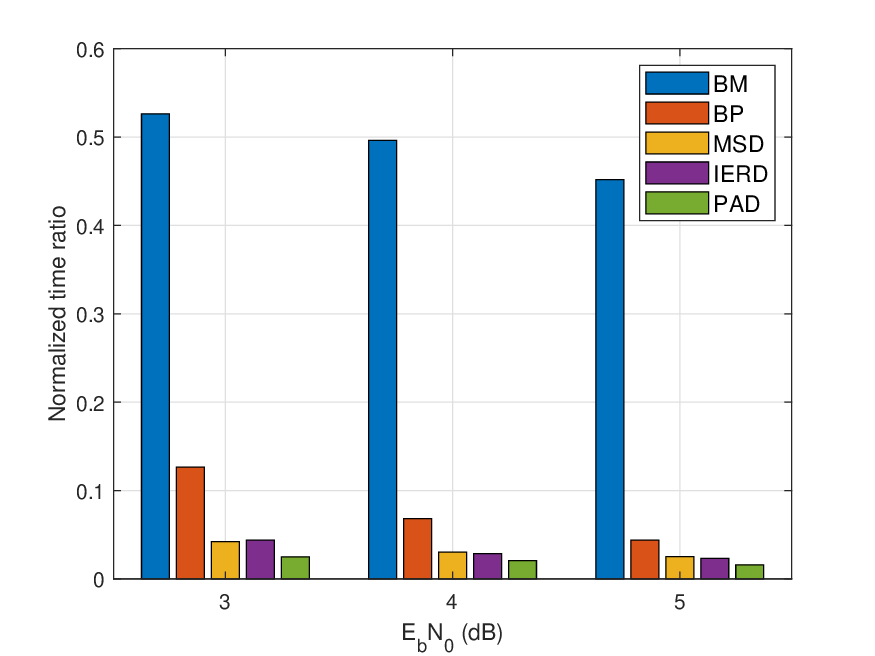} 
\caption{ Time comparison of BM, BP, MSD, IERD and PAD schemes, given 10000 blocks of 32-bit data.}
\label{fig7}
\end{figure}

\section{Conclusion}
In this work, we have studied the linear block code decoding with the aid of dual codewords within the framework of linear vector operations. Theoretical analysis has shown that dual codewords with arbitrary weights can be leveraged to extract the essential information required for decoding.  Numerical results demonstrate that the proposed method in this article can achieve hyper reliability and lower complexity performances, which makes it an ideal solution to the hyper reliability requirement of 6G xURLLC. 


\appendices
\section{Proof of Theorem \ref{thm1}}
In this section, we prove that \( W(\delta, \tau) \) is monotonically increasing with respect to \( \tau \) when \( 2\tau +2 +(\sqrt{\tau}+1)(\delta-3) \leq N \). Define
\begin{align}
M(\delta, \tau) = \sum_{\substack{k=1 \\ k \text{ odd}}}^{\tau} \binom{\tau}{k} \binom{N-\tau}{\delta-k}.
\end{align}
Thus, we express \( W(\delta, \tau) \) as $W(\delta, \tau) = \frac{M(\delta, \tau)}{\binom{N}{\delta}}$.

To establish the monotonicity of \( W(\delta, \tau) \) with respect to \( \tau \), it suffices to show that $M(\delta, \tau + 1) > M(\delta, \tau)$, which is equivalent to proving
\begin{align}
&\Delta M = M(\delta, \tau+1) - M(\delta, \tau) \notag\\
 & = \sum_{\substack{k=0 \\ k \text{ even}}}^{\tau} \binom{\tau}{k} \binom{N-\tau-1}{\delta-k-1}-\sum_{\substack{k=1 \\ k \text{ odd}}}^{\tau} \binom{\tau}{k} \binom{N-\tau-1}{\delta-k-1} \notag\\
& =\sum_{\substack{k=0}}^{\tau} (-1)^k\binom{\tau}{k} \binom{N-\tau-1}{\delta-k-1}> 0.
\end{align}
Define the generating function $ G(x) = (1-x)^\tau (1+x)^{N-\tau-1},$ then $\Delta M$ corresponds to the coefficient of $x^{\delta-1}$ in $G(x)$. By factorizing $G(x)$, we obtain $G(x) = (1-x)^\tau (1+x)^{N-\tau-1} = (1-x^2)^\tau (1+x)^{N-2\tau-1}$. Consequently, the coefficient of $x^{\delta-1}$ in $G(x)$ can be expressed as
\begin{align}
\Delta M &= \sum_{k=0}^{\tau} (-1)^k \binom{\tau}{k} \binom{N-\tau-1}{\delta-k-1} \notag\\
&= \sum_{m=0}^{\lfloor (\delta-1)/2 \rfloor} (-1)^m \binom{\tau}{m} \binom{N-2\tau-1}{\delta-1-2m}.
\end{align}

For $\delta - 1 - 2m \leq \delta - 1 \leq \frac{N-2\tau-1}{2} \quad (\text{when } \tau > 1)$, the binomial coefficient $\binom{N-2\tau-1}{\delta-1-2m}$ decreases as $m$ increases. Define the absolute value of the $m$-th term in $\Delta M$ as $a_m = \binom{\tau}{m} \binom{N-2\tau-1}{\delta-1-2m}$. Thus, to establish that $\Delta M > 0$, it suffices to show that 
\begin{align}
\label{eq:m1}
&\frac{a_{m}}{a_{m+1}} \notag= \frac{\binom{\tau}{m} \binom{N - 2\tau -1}{\delta -1 -2m}}{\binom{\tau}{m+1} \binom{N - 2\tau -1}{\delta -1 -2(m+1)}}\notag\\
&=\frac{(m+1)(N - 2\tau - \delta + 2m + 2)(N - 2\tau - \delta + 2m + 1)}{(\tau - m)(\delta -1 -2m)(\delta -2 -2m)}\\
&>\frac{a_{0}}{a_{1}}=\frac{(N - 2\tau - \delta  + 2)(N - 2\tau - \delta+ 1)}{(\tau - 1)(\delta -3)(\delta -4)}\notag \\
&>\frac{(N - 2\tau - \delta+ 1)^2}{\tau(\delta-3)^2}\geq 1  \notag\\
\label{eq:m2}
&\Leftrightarrow (N - 2\tau - \delta+ 1)^2 - \tau(\delta-3)^2 \geq 0. 
\end{align}

By solving the quadratic inequality (\ref{eq:m2}) for \(N\), we establish that under the condition \(N \geq 2\tau +2 +(\sqrt{\tau}+1)(\delta-3)\), the inequality \((N - 2\tau - \delta + 1)^2 - \tau(\delta-3)^2 \geq 0\) holds true.
Thus, when $2\tau +2 +(\sqrt{\tau}+1)(\delta-3) \leq N$, we obtain
$\frac{a_m}{a_{m+1}} >\frac{(N - 2\tau - \delta+ 1)^2}{\tau(\delta-3)^2}\geq 1$, which implies $\Delta M > 0$. This confirms that $W(\delta, \tau)$ is monotonically increasing with respect to $\tau$ under the constraints $2\tau +2 +(\sqrt{\tau}+1)(\delta-3) \leq N$. 


\begin{bibliographystyle}{IEEEtran}
\begin{bibliography}{ref}
\end{bibliography}
\end{bibliographystyle} 

\end{document}